\documentclass[twocolumn]{autart}
\usepackage{natbib}            % you should have natbib.sty
\usepackage{graphicx}
\usepackage{amsmath}
\usepackage{booktabs}
\usepackage{color}
\usepackage{amsfonts}
\usepackage{subfigure}
\usepackage{epstopdf}
\usepackage{mdwlist}
\usepackage{amssymb}
\usepackage{tabularx,ragged2e,booktabs,caption}
\newtheorem{Theorem}{Theorem}

\newtheorem{Lemma}{Lemma}
\newtheorem{Proposition}{Proposition}

\newtheorem{Remark}{Remark}
\newtheorem{Definition}{Definition}
\newenvironment{proof}{\vspace{1ex}\noindent{\bf Proof}\hspace{0.5em}}
	{\hfill\qed\vspace{1ex}}

\def\sqw{\hfill\hbox{\lower.1ex\hbox{$\sqcup$} \kern-1.02em\lower.1ex\hbox{$\sqcap$}}\ }

\DeclareMathOperator*{\argmin}{\arg\min}

\DeclareMathOperator*{\diag}{diag}

\newcommand{\yv}{\mathbf{y}}

\newcommand{\xv}{\mathbf{x}}

\newcommand{\hv}{\mathbf{h}}

\newcommand{\vv}{\mathbf{v}}

\newcommand{\AM}{\mathbf{A}}

\newcommand{\DM}{\mathbf{D}}

\newcommand{\WM}{\mathbf{W}}
\newcommand{\UM}{\mathbf{U}}
\newcommand{\VM}{\mathbf{V}}

\newcommand{\PP}{\mathbb{P}}

\newcommand{\R}{\mathbb{R}}

\newcommand{\N}{\mathbb{N}}

\begin{document}

\begin{frontmatter}

\title{Sparse Estimation From Noisy Observations\\  of an Overdetermined Linear System}

\author[First]{Liang Dai}
\author[First]{and Kristiaan Pelckmans}

\address[First]{Division of Systems and Control,
	Department of Information Technology, \\
	Uppsala University, Sweden \\
	e-mail: liang.dai@it.uu.se, kristiaan.pelckmans@it.uu.se.}

\maketitle

\begin{abstract}
	This note studies a method for the estimation of a finite number of unknown parameters from linear equations, which are perturbed by Gaussian noise.
	In case the unknown parameters have only few nonzero entries, the proposed estimator performs more efficiently than a traditional approach.
	The method consists of three steps:
	(1) a classical Least Squares Estimate (LSE),
	(2) the support is recovered through a Linear Programming (LP) optimization problem which can be computed using a soft-thresholding step,
	(3) a de-biasing step using a LSE on the estimated support set.
	The main contribution of this note is a formal derivation of an associated ORACLE property of the final estimate.
	That is, with probability 1, the estimate equals the LSE based on the support of the {\em true} parameters when the number of observations goes to infinity.
\end{abstract}
\date{}
\end{frontmatter}
\begin{section}{Problem settings}

This note considers the estimation of a sparse parameter vector from noisy observations of a linear system.
The formal definition and assumptions of the problem are given as follows.
Let $n>0$ be a fixed number, denoting the dimension of the underlying parameter vecto, and let $N>0$ denote the number of equations ('observations').
The observed signal $\yv\in\R^N$ obeys the following system:
\begin{equation}
 \yv = \AM \xv^0 + \vv,
 \label{eq.mdl}
\end{equation}
where the elements of the vector $\xv^0\in \R^{n}$ are considered to be the fixed but unknown parameters of the system.
Moreover, it is assumed that $\xv_0$ is $s$-sparse (i.e. there are $s$ nonzero elements in the vector).
Let $\mathcal{T}\subset\{1, \dots, n\}$ denote the support set of $\xv^0$ (i.e. $\xv^0_{i}=0\Leftrightarrow i\not\in \mathcal{T}$) and ${\mathcal{T}}^{c}$ be the complement of $\mathcal{T}$, i.e. $\mathcal{T}\bigcup {\mathcal{T}}^{c} = \{1,2,\cdots, n\}$ and $\mathcal{T}\bigcap {\mathcal{T}}^{c} = \emptyset$.
The elements of the vector $\vv \in \R^{N}$ are assumed to follow the following distribution
\begin{equation}\vv \sim \mathcal{N} (0,c I_{N}),\end{equation}
where $0<c\in\mathbb{R}$.

Applications of such setup appear in many places, to name a few, see the applications discussed in \cite{c9} on the detection of nuclear material, and in \cite{c10} on model selection for aircraft test modeling (see also the Experiment 2 in \cite{c1} on the model selection for the AR model). In the experiment section, we will demonstrate an example which finds application in line spectral estimation, see \cite{c14}.

The matrix $\AM \in \R^{N\times n}$ with $N> n$ is the sensing matrix.
Such a setting ($\AM$ is a 'tall' matrix) makes it different from the setting studied in compressive sensing,
where the sensing matrix is 'fat', i.e. $N \ll n$. For an introduction to the compressive sensing theory, see e.g. \cite{c2,c16}.

 Denote the Singular Value Decomposition (SVD) of matrix $\AM\in\R^{N\times n}$ as
 \begin{equation}
	\AM = \UM\Sigma \VM^{T},
 	\label{eq.svd}
 \end{equation}
in which $\UM \in \R^{N\times n}$ satisfies $\UM ^{T} \UM = I_n$,
 $\VM \in \R^{n\times n}$ satisfies $\VM ^{T} \VM = I_n$,
and $\Sigma\in \R^{n\times n}$ is a diagonal matrix $\Sigma = \diag(\sigma_1(\AM), \sigma_2(\AM), \ldots, \sigma_n(\AM))$.
The results below make the following assumptions on $\AM$:
\begin{Definition}
	We say that $\{\AM\in\R^{N\times n}\}_N$ are sufficiently rich if there exists a finite $N_0$ and  $0<c_1\le c_2$
	such that for all $N>N_0$ the corresponding matrices $\AM\in\R^{N\times n}$ obey
	\begin{equation}
    	c_1\sqrt{N}\le \sigma_1(\AM)\le\sigma_2(\AM)\le\ldots\le\sigma_n(\AM)\le c_2\sqrt{N},
    		\label{eq.sing}
 	\end{equation}
	where $\sigma_i(\AM)$ denotes the $i$-th singular value of the matrix $\AM$, $c_1, c_2 \in \mathbb{R}^{+}$.
 \end{Definition}
Note that the dependence of $\AM$ on $N$ is not stated explicitly in order to avoid notational overload.

 In \cite{c1} and \cite{c12}, the authors make the assumption on $\AM$ that the sample covariance matrix
 $\frac{1}{N}{\AM}^{T}\AM$
 converges to a finite, positive-definite matrix:
 \begin{equation}
	 \lim_{N\rightarrow\infty} \frac{1}{N}{\AM}^{T}\AM = \DM\succ 0.
     \label{eq.sing2}
 \end{equation}
This assumption is also known as {\em Persistent Excitation} (PE), see e.g. \cite{c6}.
Note that our assumption in Eq. (\ref{eq.sing}) covers a wider range of cases.
For example, Eq. (\ref{eq.sing}) does not require the singular values of $\frac{1}{\sqrt{N}}{\AM}$ to converge, while only requires that they lie in  $[c_1,c_2]$ when $N$ increases.

Classically, properties of the Least Square Estimate (LSE) under the model given in Eq. (1) are given by the Gauss-Markov theorem.
It says that the Best Linear Unbiased Estimation (BLUE) of $\xv^{0}$ is the LSE under certain assumptions on the noise term. For the Gauss-Markov theorem, please refer to \cite{c3}.
However, the normal LSE does not utilize the 'sparse' information of $\xv^{0}$, which raises the question that whether it is possible to improve on the normal LSE by exploiting this information. In the literature, several approaches have been suggested, which can perform as if the true support set of $\xv^{0}$ were known. Such property is termed as the ORACLE property in \cite{c7}. In \cite{c7}, the SCAD (Smoothly Clipped Absolute Deviation) estimator is presented, which turns out to solve a non-convex optimization problem; later in \cite{c12}, the ADALASSO (Adaptive Least Absolute Shrinkage and Selection Operator) estimator is presented. The ADALASSO estimator consists of two steps, which implements a normal LSE in the first step, and then solves a reweighed Lasso optimization problem, which is convex. Recently, in \cite{c1}, two LASSO-based estimators, namely the 'A-SPARSEVA-AIC-RE' method and the 'A-SPARSEVA-BIC-RE' method, are suggested. Both methods need to do the LSE in the first step, then solve a Lasso optimization problem, and finally redo the LSE estimation.

 \begin{Remark}
  This note concerns the case that $\xv^0$ is a fixed sparse vector. However, when sparse estimators are applied to estimate non-sparse vectors, erratic phenomena could happen. For details, please see the discussions in \cite{c15,c19}.
 \end{Remark}

In this note, we will present another approach to estimate the sparse vector $\xv^{0}$, which also possesses the ORACLE property with a lighter  computational cost. The proposed method consists of three steps, in the first step, a normal LSE is conducted, the second step is to solve a LP (Linear Programming) problem, whose solution is given by a soft-thresholding step, finally, redo the LSE based on the support set of the estimated vector from the previous LP problem. Details will be given in Section 2.

In the following, the lower bold case will be used to denote a vector and capital bold characters are used to denote matrices. The subsequent sections are organized as follows.
In section 2, we will describe the algorithm in detail and an analytical solution to the LP problem is given.
In Section 3, we will analyze the algorithm in detail.
In Section 4, we conduct several examples to illustrate the efficacy of the proposed algorithm and compare the proposed algorithm with other algorithms.
Finally, we draw conclusions of the note.
\end{section}

\begin{section}{Algorithm Description}
%%%%%%%%%%%%%%%%%%%%%

The algorithm consists of the following three steps:
\begin{itemize}

\item{\em LSE:}
	Compute the LSE of $\xv^{0}$, denoted as $\xv^{ls}$.
\item{\em LP:}
	Choose $0<\epsilon<1$ and solve the following Linear Programming problem:
	\begin{equation}
		\xv^{lp} = \argmin_\xv \| \xv \|_1 \mbox{\ s.t. \ }  \| \xv-\xv^{ls} \|_{\infty} \le \lambda,
		\label{eq.lp}
	\end{equation}
	where $\lambda = \sqrt{\frac{2n}{N^{1-\epsilon}}}$.
	Detect the support set ${\mathcal{T}}^{lp}$ of $\xv^{lp}$.
\item{\em RE-LSE:}
	Compute the LSE of $\xv^{0}$ based on ${\mathcal{T}}^{lp}$. Form the matrix $\AM_{{\mathcal{T}}^{lp}}$, which contains the columns of $\AM$ indexed by ${\mathcal{T}}^{lp}$ and let $\AM^{\dagger}_{{\mathcal{T}}^{lp}}$ denote its pseudo-inverse.
	Then the final estimation $\xv^{rels}$ is given by
 $\xv_{{\mathcal{T}}^{lp}}^{rels} = \AM_{{\mathcal{T}}^{lp}}^{\dagger}\yv$, and $\xv_{{{\mathcal{T}}^{lp}}^{C}}^{rels} = \mathbf{0}$, in which ${{\mathcal{T}}^{lp}}^{C}$ denotes the complement set of ${\mathcal{T}}^{lp}$.

\end{itemize}
Note that the LP problem has an analytical solution. Writing the $\infty$ norm constraint explicitly as
 \begin{align}
	&\xv^{lp} = \argmin_\xv \sum_{i=1}^{n}|x_i| \\ \nonumber
    &\mbox{\ s.t. \ }  | x_i-x^{ls}_{i} | \le \lambda, \mbox{\ for \ } i = 1 \ldots n.
    \label{eq.lp2}
\end{align}
We can see that there are no cross terms in both the objective function and the constraint inequalities, so each component can be optimized separately. From this observation, the solution of the LP problem is given as
\begin{equation}
x^{lp}_i =
\begin{cases}
 0,                       &\mbox{\ if \ }  |x^{ls}_{i}| \le \lambda\\ \nonumber
 x^{ls}_{i} - \lambda,  &\mbox{\ if \ }  x^{ls}_{i} > \lambda\\ \nonumber
 x^{ls}_{i} + \lambda,  &\mbox{\ if \ }  x^{ls}_{i} < -\lambda \nonumber
 \end{cases}
\end{equation}for $i = 1,2,\cdots,n$.
 Such a solution $\xv^{lp}$ is also referred to as an application of the soft-thresholding operation to $\xv^{ls}$, see e.g. \cite{c5}. Several remarks related to the algorithm are given as follows.
\begin{Remark}
	Note that the tuning parameter $\lambda$ chosen as $\lambda^2 = \frac{2n}{N^{1-\epsilon}}$
	is very similar to  the one (which is proportional to $\frac{2n}{N}$) as given in \cite{c1} based on the Akaike's Information Criterion (AIC).
\end{Remark}
\begin{Remark}
The {\em order} of $\lambda$ chosen as $-\frac{1}{2}+\frac{\epsilon}{2}$  is essential to make the asymptotical oracle property hold.  Intuitively speaking, such a choice can make the following two facts hold.
\begin{enumerate}
\item
Whenever $\epsilon > 0$, $\xv^0$ will lie in the feasible region of Eq. (\ref{eq.lp}) with high probability.
\item
The threshold decreases 'slower' (in the order of $N$) than the variance of the pseudo noise term $\VM\Sigma^{-1}\UM^{T}\vv$. With such a choice, it is possible to get a good approximation of the support set of $\xv^0$ in the second step.
\end{enumerate}

\end{Remark}

\begin{Remark}
Though the formulation of Eq. (\ref{eq.lp}) is inspired by the Dantzig selector in \cite{c8}, there are some differences between them.
\begin{enumerate}
 \item
 As pointed out by one of the reviewer, both the proposed method and the Dantzig selector lie in the following class
\begin{equation}
  \min_\xv \|\xv\|_1 \\
    \mbox{\ s.t. \ }  \|\mathbf{W} (\xv-\xv^{ls})\|_{\infty} \le \lambda.
    \label{eq.lp3}
\end{equation}
If $\mathbf{W}$ is chosen as the identity matrix, we obtain the proposed method; If $\mathbf{W}$ is chosen as $\AM^{T}\AM$, then we obtain the same formulation as given by the Dantzig selector.
%However, the Dantzig selector is designed for compressive sensing problems, in which case $N\ll n$ is assumed, and the ORACLE property for the Dantzig selector is discussed under a different assumption of the sensing matrix $\AM$, which is called the RIP (Restricted Isometry Property)  condition. Also, in the Compressive sensing problems, the number of rows (namely $N$ here) of the sensing matrix $\AM$ is fixed, while here, we consider the case where $N$ will grow to infinity.

\item
 As pointed out in \cite{c13}, the solution path of the Dantzig selector behaves erratically with respect to the value of the regularization parameter. However, the solution path of Eq. (\ref{eq.lp}) with respect to the value of $\lambda$ behaves regularly, which is due to the fact that, given $\lambda$, the solution to Eq. (\ref{eq.lp}) is given by the application of the soft-thresholding operation to the LSE estimation. When $\lambda$ increases, the solution will decrease (or increase) linearly and when it hits zero, it will remain to be zero. This in turn implies computational advances when trying to find a $s$-sparse solution for given $s$.  A simple illustration of the solution path is given. Assume that $n = 4$ and $\xv^{ls} = [2,0.5,-1,-1.5]^{T}$, then the solution path to Eq. (\ref{eq.lp}) w.r.t. $\lambda$ is given in Fig. \ref{fig.path}).
\begin{figure}[htbp]
\centering
\includegraphics[width = 75mm]{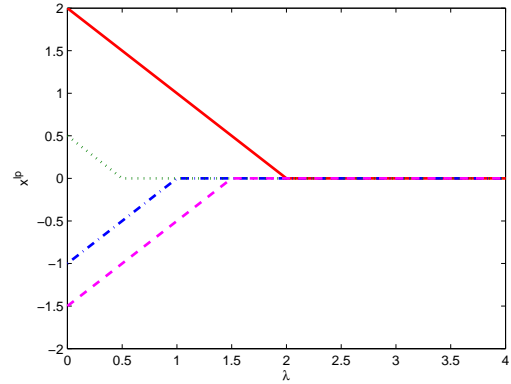}
  \caption{An illustration of the solution path to Eq. (\ref{eq.lp}) w.r.t. $\lambda$. When $\lambda$ equals zero, the solution to Eq. (\ref{eq.lp}) is $\xv^{ls}$; when $\lambda$ increases, the solution trajectory shrinks linearly to zero and then remains zero.}
\label{fig.path}
\end{figure}
\end{enumerate}
\end{Remark}
\begin{Remark}
From a computational point of view, the SCAD method needs to solve a non-convex optimization problem which will suffer from the multiple local minima, see the discussions in \cite{c18}. Hence, the proposed scheme is mainly compared with techniques which can be solved as convex optimization problems. In Table \ref{comp.table}, we list the computational steps needed for different methods. In the table, the term ST means the soft-thresholding operation, the term Re-LSE means 'redo the LSE estimation after detecting the support set of the result obtained from the second step'. For a more precise description, see the Algorithm Description section. From this table, we can see that in the first step, all the methods need to do a LSE estimation; in the second step, except the proposed method (which is denoted by LP + Re-LSE), the other methods need to solve a LASSO optimization problem, which is more computationally involved than a simple soft-thresholding operation as needed by the proposed method; except the ADALASSO method, the other methods need to do a Re-LSE step.  From this table, we can also see that the main computational burden for the proposed method comes from the LSE step.
\begin{table}
\caption{Computational steps needed for different methods}
    \begin{tabular}{|l|l|l|l|}
    \hline
    ~                 & Step 1 & Step 2 & Step 3 \\ \hline
    LP + Re-LSE       & LSE    & ST     & Re-LSE \\ \hline
    ADALASSO          & LSE    & LASSO  & ~      \\ \hline
    A-SPARSEVA-AIC-RE & LSE    & LASSO  & Re-LSE \\ \hline
    A-SPARSEVA-BIC-RE & LSE    & LASSO  & Re-LSE \\ \hline
    \end{tabular}
    \label{comp.table}
\end{table}
\end{Remark}

\begin{Remark}
Note that the proposed method does not need an "adaptive step" (i.e. to reweigh the cost function) in order to achieve the ORACLE property, which is different from the methods presented in \cite{c1} and \cite{c12}.
\end{Remark}

\end{section}

\begin{section}{Analysis of the algorithm}
In this section, we will discuss the properties of the presented estimator. In the following, we will denote the smallest singular value of $\AM$ as $\sigma$.

\begin{Remark}

In the following sections, we assume that the noise variance equals one, i.e. $c=1$, for the following reasons:
\begin{enumerate}
\item
When the noise variance is given in advance, one can always re-scale the problem accordingly.
\item
 Even if the noise variance is not known explicitly (but is known to be finite),  the support of $\xv^{0}$ will be recovered asymptotically.
This is a direct consequence of the fact that finite, constant scalings do not affect the asymptotic statements, i.e. we can use the same $\lambda$ for any level of variance without influencing the asymptotic behavior.
\end{enumerate}
\end{Remark}

The following facts (Lemma 1-3) will be needed for subsequent analysis.
Since their proofs are standard, we state them without proofs here.
Using the notations as introduced before, one has that
\begin{Lemma}
$\xv^{ls} = \xv^0 + \VM\Sigma^{-1}\UM^{T}\vv$.
\end{Lemma}

\begin{Lemma}
$\mathbf{b} = \Sigma \VM^{T}\xv^{ls} - \Sigma \VM^{T}\xv^0$ is a Gaussian random vector with distribution $\mathcal{N}(0, I)$.
\end{Lemma}

\begin{Lemma}
Given $d>0$, then \[ \int_{|t|>d}\frac{1}{\sqrt{2\pi}}e^{-\frac{t^2}{2}}dt \le e^{-\frac{d^2}{2}}. \]
\end{Lemma}

In the following, we will first analyze the probability that $\xv^0$ lies in the constraints set of the LP problem given by Eq. (\ref{eq.lp}).
Then we give an error estimation of the results given by Eq. (\ref{eq.lp}).
After this, we will discuss the capability of recovering the support set of $\xv^{0}$ by Eq. (\ref{eq.lp}),
which will lead to the asymptotic ORACLE property of the proposed estimator.
\begin{Lemma}
For all $\lambda > 0$, one has that
\[\PP\left(\|\VM^{T}\xv^{ls} - \VM^{T}\xv^0\|_{\infty}> \frac{\lambda}{\sqrt{n}}\right) \le ne^{-\frac{\lambda^2\sigma^2}{2n}}.\]
\label{lm.prob}
\end{Lemma}
\begin{proof}
By Lemma 2, and noticing that $\mathbf{b} = \Sigma \VM^{T}\xv^{ls} - \Sigma \VM^{T}\xv^0 $ is a Gaussian random vector with distribution $\mathcal{N}(0, I)$, we have that
\begin{multline}
 \PP\left(\|\VM^{T}\xv^{ls} - \VM^{T}\xv^0\|_{\infty}> \frac{\lambda}{\sqrt{n}}\right)  \\
\le
\PP\left(\|\Sigma \VM^{T}\xv^{ls} - \Sigma \VM^{T}\xv^0\|_{\infty}> \frac{\lambda \sigma}{\sqrt{n}}\right)  \\ \nonumber
=  \PP\left(\|\mathbf{b}\|_{\infty}> \frac{\lambda \sigma}{\sqrt{n}}\right)\\ \nonumber
= \PP\left(\exists i, such\ \ that\ \ |b_i|> \frac{\lambda \sigma}{\sqrt{n}}\right)\\ \nonumber
\le \sum_{i=1}^{i=n}\PP\left({|b_i|> \frac{\lambda \sigma}{\sqrt{n}}}\right).\nonumber
\end{multline}
Application of Lemma 3 gives the desired result.
\end{proof}

\begin{Lemma}
For all $\lambda > 0$, if $\|\VM^{T}\xv^{ls} - \VM^{T}\xv^0\|_{\infty} \le \frac{\lambda}{\sqrt{n}}$, then $\|\xv^{ls} - \xv^0\|_{\infty} \le \lambda$.
\end{Lemma}
\begin{proof}
Define  $\mathbf{c}$ as  $\mathbf{c} = \VM^{T}\xv^{ls} - \VM^{T}\xv^0 $, so we have $\|\xv^{ls} - \xv^0\|_{\infty} = \|V\mathbf{c}\|_{\infty}$.
Analyze the $i$th element of $\VM \mathbf{c}$ that
\begin{align}\nonumber
 |{\VM}_i\mathbf{c}| \le \|\mathbf{c}\|_2
               \le \|\mathbf{c}\|_{\infty}\sqrt{n} \le \lambda. \nonumber
\end{align}
The first inequality is by definition, the second inequality comes from the Cauchy inequality, the last inequality is due to the assumption of the lemma.
\end{proof}

Combining the previous two lemmas gives
\begin{Lemma}
	$\PP(\|\xv^{ls} - \xv^0\|_{\infty} \le \lambda) \ge 1 - ne^{-\frac{\lambda^2\sigma^2}{2n}}$.
\end{Lemma}
\begin{proof}
The proof goes as follows
\begin{multline}
\PP\left(\|\xv^{ls} - \xv^0\|_{\infty} \le \lambda\right) \\
\ge \PP\left(\|\VM^{T}\xv^{ls} - \VM^{T}\xv^0\|_{\infty}\le \frac{\lambda}{\sqrt{n}}\right) \\
= 1 - \PP\left(\|\VM^{T}\xv^{ls} - \VM^{T}\xv^0\|_{\infty} > \frac{\lambda}{\sqrt{n}}\right) \\
\ge 1 - ne^{-\frac{\lambda^2\sigma^2}{2n}} \nonumber
\end{multline}
The first inequality comes from Lemma 5, and the second inequality follows from Lemma 4.
\end{proof}

The above lemma tells us that $\xv^0$ will lie inside the feasible set of the LP problem as given in Eq. ($\ref{eq.lp}$) with high probability.
By a proper choice of $\lambda$, the following result is concluded.
\begin{Theorem}
	 Given $0<\epsilon < 1$, and let $\lambda^2 = \frac{2n}{N^{1-\epsilon}}$, we have that
 $$\PP\left(\|\xv^{ls} - \xv^0\|_{\infty} \le \lambda\right) \ge 1 - ne^{-c_1^2 N^{\epsilon}}.$$
\end{Theorem}

Next, we will derive an error bound (in the $l_2$- norm) of the estimator given by the LP formulation.
Define
 $$\hv = \xv^{lp}-\xv^0,$$
 as the error vector of LP formulation.
 We have that  the error term $\hv$ is bounded as follows:
  \begin{Lemma}
 	For any $\lambda>0$,
	if $\|\xv^{ls} - \xv^0\|_{\infty} \le \lambda$,
 	then we have that $\| \hv \|_2^2\le 4s\lambda^2$.
 \end{Lemma}
 \begin{proof}
	We first consider the error vector on ${\mathcal{T}}^{c}$ which is given by $\hv_{{\mathcal{T}}^{c}}$.
	Since $\|\xv^{ls} - \xv^0\|_{\infty} \le \lambda$ and $\xv^0_{{\mathcal{T}}^{c}} = \mathbf{0}$,  we have that $\|\xv^{ls}_{{\mathcal{T}}^{c}}\|_{\infty} \le \lambda$.
	It follows from the previous discussions that $\xv^{lp}$ is obtained by application of the soft-shresholding operator with the threshold $\lambda$, 	 applied  componentwise to $\xv^{ls}$, hence we obtain that $\xv_{{\mathcal{T}}^{c}}^{lp} = \mathbf{0}$.
	This implies that $\hv_{{\mathcal{T}}^{c}} = \mathbf{0}$.

	Next we consider the error vector on the support $\mathcal{T}$, denoted as $\hv_{\mathcal{T}}$.
	From the property of the soft-thresholding operation, it follows that
	$\|\xv_{\mathcal{T}}^{ls} - \xv_{\mathcal{T}}^{lp}\|_{\infty} \le \lambda.$
	Then we have that $\|\xv_{\mathcal{T}}^{0} - \xv_{\mathcal{T}}^{lp}\|_{\infty} \le \|\xv_{\mathcal{T}}^{ls} - \xv_{\mathcal{T}}^{lp}\|_{\infty} + \|\xv_{\mathcal{T}}^{ls} - \xv_{\mathcal{T}}^{0}\|_{\infty} \le 2\lambda$

	Combining both statements gives that $\|\hv\|_2^2 = \|\hv_{\mathcal{T}}\|_2^2 + \|\hv_{{\mathcal{T}}^{c}}\|_2^2 \le |T|\|\hv_{\mathcal{T}}\|_{\infty}^2 \le 4s\lambda^2 $.
\end{proof}

Plugging in the $\lambda$ as chosen in previous section, we can get the error bound of the LP formulation.
 However, the estimate $\xv^{lp}$ is not the final estimation, instead it will be used to recover the support set of $\xv^{0}$.
The following theorem states this result formally.
For notational convenience, ${\mathcal{T}}^{lp}(N)$ is used to denote the recovered support from the LP formulation,
and $\xv^{rels}(N)$ then denotes the estimate after the second LSE step using $N$ observations.
Finally, the vector $\xv^{ls-or}(N)$ denotes the LSE as if the support of $\xv^0$ were known (i.e. the ORACLE presents) using $N$ observations.

We will first get a weak support recovery result and based on this, we further prove that the support as recovered by the LP formulation will converge to the true support $\mathcal{T}$ with probability 1 when $N$ goes to infinity.
\begin{Lemma}
	Given $0<\epsilon < 1$, and assume that the matrix $\AM$ has singular values which satisfies Eq. (\ref{eq.sing}), with constants $c_1, c_2$ as given there. Let $x_0  \triangleq  \min \{ |x^0_i|, i \in \mathcal{T} \}\in \R^{+}$, and $\lambda^2 = \frac{2n}{N^{1-\epsilon}}$, then
	\begin{equation}
	    \lim_{N\rightarrow\infty}\PP({\mathcal{T}} = {\mathcal{T}}^{lp}(N)) = 1. \nonumber
	\end{equation}
\end{Lemma}
\begin{proof}
	Let the vector $\bar\vv$ denote $\bar{\vv} = \VM\Sigma^{-1}\UM^{T}\vv$.
 	Since $\xv^{ls} = \xv^{0} + \VM\Sigma^{-1}\UM^{T}\vv$, one has that $\xv^{ls} = \xv^{0} + \bar{\vv}$, in which $\bar{\vv}$ follows a normal distribution $\mathcal{N}(0,\VM\Sigma^{-2}\VM^{T})$.
	Without loss of generality, assume that $x_1^{0},x_2^{0},\dots,x_s^{0}$ are the nonzero elements of $\xv^{0}$ and their values are positive.
    Since $\lambda$ decreases when $N$ increases, so there exist a number $N_1 \in \N$, such that $\lambda < \frac{x_0}{2}$ for all $N\ge N_1$. In the following derivations, we use $v_{i,j}$ to denote the element in the $i$th row, $j$th column of $\VM$ and $\bar{v}_i$ denotes the $i$th element of $\bar{\vv}$.
    When $N > N_1$, we have the following bound of $\PP(\mathcal{T}\ne {\mathcal{T}}^{lp}(N))$:
	 \begin{align}\nonumber
 	&\PP\left(\mathcal{T} \ne {\mathcal{T}}^{lp}(N)\right)\\ \nonumber
	&= \PP \left(|x_1^{0}+\bar{v}_1|<\lambda, or \ \ |x_2^{0}+\bar{v}_2|<\lambda, \dots, or \ \ |x_s^{0}+\bar{v}_s|<\lambda; \right . \\ \nonumber
 	& \ \ \ \ \ \ \ \ \  \ \ \ \ \ \ \ \left. or \ \ |\bar{v}_{s+1}|>\lambda, or \ \ |\bar{v}_{s+2}|>\lambda,  \dots ,or \ \ |\bar{v}_{N}|>\lambda \right)\\ \nonumber
 	&\le \sum_{i=1}^{s}\PP (-\lambda-x_i^0<\bar{v}_i<\lambda-x_i^0) + \sum_{i=s+1}^{N} \PP (|\bar{v}_i|>\lambda) \\ \nonumber
 	&\le \sum_{i=1}^{s}\frac{2\lambda e^{-(2\sum_{j=1}^n\sigma_j^{-2}v_{ij}^2)^{-1}(-x_i^0+\lambda)^2}}{\sqrt{2\pi(\sum_{j=1}^n\sigma_j^{-2}v_{ij}^2)}}  \nonumber \\
	& \hspace{+4cm}  + \sum_{i=s+1}^{N}e^{-(2\sum_{j=1}^n\sigma_j^{-2}v_{ij}^2)^{-1}\lambda^2} \nonumber \\
 	&\le \sum_{i=1}^{s}\frac{2c_2\sqrt{N}\lambda}{\sqrt{2\pi}} e^{-\frac{1}{2}c_1^2N(-x_i^0+\lambda)^2} + \nonumber \sum_{i=s+1}^{N}e^{-\frac{1}{2}c_1^2N\lambda^2}\\ \nonumber
 	&\le 2c_2s\sqrt{n}N^{\frac{\epsilon}{2}} e^{-\frac{1}{8}(c_1x_0)^2 N} + N e^{-c_1^2nN^{\epsilon}} \nonumber \\
   & =  C N^{\frac{\epsilon}{2}} e^{-\frac{1}{8}(c_1x_0)^2 N} + N e^{-c_1^2nN^{\epsilon}},
   \label{eq.pbconst}
	\end{align}
	where $C = 2c_2s\sqrt{n}$. The second inequality in the chain holds due to the fact that the probability distribution function of $\bar{v}_i$ is monotonically increasing in the interval $[-\lambda-x_i^0,\lambda-x_i^0]$, together with results in Lemma 3.

Then we can see that both terms in (\ref{eq.pbconst}) will tend to 0 as $N \rightarrow \infty$ for any fixed $\epsilon >0$, i.e. $\lim_{N\rightarrow\infty} \PP({\mathcal{T}}^{lp}(N) = {\mathcal{T}})=1$.
\end{proof}

\begin{Remark}
	Notice the fact that
	\begin{equation}
    		\PP\left(\xv^{rels}(N) = \xv^{ls-or}(N)\right)
    		\ge \PP\left({\mathcal{T}}^{lp}(N) = {\mathcal{T}}\right), \nonumber
	\end{equation}
	and from the previous Lemma, we know that the right hand side will tend to 1 as $N$ tends to infinity, so it also holds that
	\begin{equation}
	   \lim_{N\rightarrow\infty} \PP(\xv^{rels}(N) = \xv^{ls-or}(N)) = 1. \nonumber
	\end{equation}
\end{Remark}

Based on the previous lemma, we have
\begin{Theorem}
	Given $0<\epsilon < 1$, and assume that the matrix $\AM$ has singular values which satisfies Eq. (\ref{eq.sing}), with constants $c_1, c_2$ as given there. Let $x_0  \triangleq  \min \{ |x^0_i|, i \in \mathcal{T} \}\in \R^{+}$, and $\lambda^2 = \frac{2n}{N^{1-\epsilon}}$, then it holds that
	\begin{equation}
		\PP\left(\exists N' \text{such that\,} \cap_{N= N'}^{\infty} \{ \mathcal{T}^{lp}(N) = \mathcal{T} \} \right)= 1. \nonumber
	\end{equation}
\end{Theorem}
\begin{proof}
	From the proof in the previous lemma, we have that when $N>N_1$
	\begin{align} \nonumber
 		&\PP(\mathcal{T} \ne {\mathcal{T}}^{lp}(N))\\ \nonumber
	 	&\le C N^{\frac{\epsilon}{2}} e^{-\frac{1}{8}(c_1x_0)^2 N} + N e^{-c_1^2nN^{\epsilon}} \\ \nonumber
 		&=  C e^{-\frac{1}{8}(c_1x_0)^2 N +\frac{ \epsilon}{2} ln(N)} + e^{\ln(N)-c_1^2nN^{\epsilon}} \\ \nonumber
 		&=  C e^{(c_1x_0)^2N(\frac{\epsilon ln(N)}{2(c_1x_0)^2N}-\frac{1}{8})} + e^{c_1^2nN^{\epsilon}(\frac{\ln(N)}{c_1^2nN^{\epsilon}}-1)}.
	\end{align}
	Since $0<\epsilon<1$ and $x_0>0$, one has that $\frac{\epsilon ln(N)}{2(c_1x_0)^2N}$ and $\frac{\ln(N)}{c_1^2nN^{\epsilon}}$ will tend to zero if $N \rightarrow \infty$.
	Hence there exists a number $N_2 \in \N$ such that for all $N>N_3 \triangleq \max(N_1,N_2) $ one has that $\frac{\epsilon ln(N)}{2(c_1x_0)^2N} < \frac{1}{16}$
	and $\frac{\ln(N)}{c_1^2nN^{\epsilon}} < \frac{1}{2}$.
	Hence
	\begin{align} \nonumber
 		&\sum_{N = N_3}^{\infty}\PP( {\mathcal{T}}^{lp}(N)\ne \mathcal{T} )\\ \nonumber
 		&\le \sum_{N = N_3}^{\infty} Ce^{-\frac{1}{16}(c_1x_0)^2N} + \sum_{N = N_3}^{\infty} e^{-\frac{1}{2}c_1^2nN^{\epsilon}} \\ \nonumber
 		&\le \int_{N = N_3 - 1}^{\infty} Ce^{-\frac{1}{16}(c_1x_0)^2t} dt + \int_{N_3 - 1}^{\infty} e^{-\frac{1}{2}c_1^2nt^{\epsilon}}dt \\ \nonumber
 		& = A + B \nonumber.
	\end{align}
	Furthermore, it can be seen that $$A = \int_{N = N_3 - 1}^{\infty} C e^{-\frac{1}{16}(c_1x_0)^2t} dt < \infty .$$
	In the following, we will show that $B = \int_{N_3 - 1}^{\infty} e^{-\frac{1}{2}c_1^2nt^{\epsilon}}dt < \infty$ .
	By a change of variable using $x = \frac{1}{2}c_1^2nt^{\epsilon}$, we have that
	\begin{align}\nonumber
		B &=\frac{1}{c_1^2n\epsilon} \int_{\frac{1}{2}c_1^2n(N_3-1)^{\epsilon}}^{\infty}x^{\frac{1}{\epsilon}-1} e^{-x} dx
  		< \frac{1}{c_1^2n\epsilon} \Gamma\left(\frac{1}{\epsilon}\right)
  		< \infty \nonumber
	\end{align} with $\Gamma$ the Gamma function.
	And hence
	\begin{equation} \nonumber
		\sum_{N = N_3}^{\infty}\PP( {\mathcal{T}}^{lp}(N)\ne \mathcal{T} ) < \infty.
	\end{equation}
	Application of the Borel-Cantelli lemma [4] implies that the events in $\{\mathcal{T} \ne {\mathcal{T}}^{lp}(N)\}_{N=N_3}^{\infty}$ will not happen infinitely often, which concludes the result.
\end{proof}
\end{section}

\begin{section}{Illustrative Experiments}
%%%%%%%%%%%%%%%%%%%%

This section supports the findings in the previous section with numerical examples and make the comparisons with the other algorithms which possess the ORACLE property in the literature.
\begin{subsection}{Experiment 1}
This example is taken from \cite{c12}. The setups are repeated as follows.
\begin{itemize}
\item
$\xv^{0}$ is set to be $(3,1.5,0,0,2,0,0,0)^{T}$;
\item
 Rows of matrix $A$ are i.i.d. normal vectors;
\item
The correlation between the $j_1$-th and the $j_2$-th elements of each row are given as $0.5^{|j_1-j_2|}$;
\item
The noise term $\vv\in\R^N$ follows distribution $\mathcal{N}(0,I_{N})$.
\end{itemize}
Based on these setups, the proposed method and also the methods presented in \cite{c1} (the A-SPARSEVA-AIC-RE method and the A-SPARSEVA-BIC-RE methods) and \cite{c12} (the ADALASSO method) are applied to recover $\xv^{0}$.
In this experiment,  $\epsilon$ for the proposed method is set to $\frac{1}{3}$; $\lambda_N$ for 'ADALASSO' is chosen as $N^{1/2-\gamma/4}$ (this choice satisfies all the assumptions in Theorem 2 in \cite{c12}), and $\gamma$ is set to 1; the thresholding value (for detecting zero components from the solution of the Lasso problem) for the 'A-SPARSEVA-AIC-RE' and 'A-SPARSEVA-BIC-RE' are set to be $10^{-5}$ as suggested in \cite{c1}. For the comparison, we also include the experiment result obtained by using the LASSO method, in which we set the tuning parameter as $\sqrt{N}$. In Fig. \ref{fig.exp}, for every $N$, experiment is repeated 50 times to get the estimated MSE.
The following abbreviations are used in Fig. \ref{fig.exp}:
(1) the curve with tag 'LSE' gives the MSE of the estimates by the LSE algorithm;
(2) the curve with tag 'LP + RE-LSE' gives the MSE of the estimates given by the proposed algorithm;
(3) the curve with tag 'ORACLE-LSE' gives the MSE of the estimates by the ORACLE LSE;
(4) the curves with tags 'A-SPARSEVA-AIC-RE' and 'A-SPARSEVA-BIC-RE' give the MSE of the estimates by the methods presented in \cite{c1};
(5) the curve with tag 'ADALASSO' gives the MSE of the estimates by the ADALASSO method presented in \cite{c12};
(6) the curve with tag 'LASSO' gives the MSE of the estimates of the LASSO method.

Note that, when $N$ becomes large, the curves 'LP + RE-LSE' and 'ORACLE-LSE' exactly match each other.

\begin{figure}[htbp]
\centering
\includegraphics[width = 80mm]{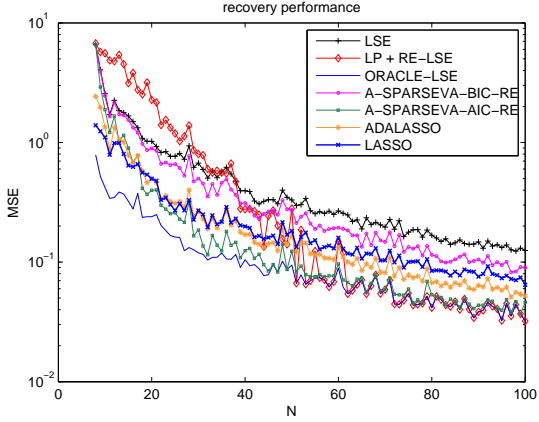}
  \caption{
  	  	Performance of the different estimators from $N$ observations to estimate $\xv^{0}$. This picture indicates that the proposed estimator will give exactly the same performance as the ORACLE estimator for a large $N$ $(N\approx75)$. }
\label{fig.exp}
\end{figure}
Fig. \ref{fig.support} demonstrates the efficacy of support recovery of the LP formulation in Eq. (\ref{eq.lp}) for different choices of $\epsilon$.
In the plot, 'portion' is defined as the ratio of successful trials over the total number of trials.
We conclude the empirical observations for this experiment in the caption of the figure.

\begin{figure}[htbp]
\centering
\includegraphics[width = 80mm]{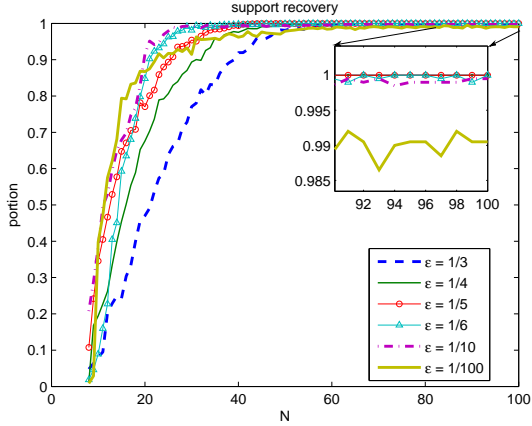}
  \caption{Support recovery performance of Eq. (\ref{eq.lp}) for different choices of $\epsilon$. Empirically, we observe that:
  1) When $\epsilon$ is chosen to be small, the ratio for successful support recovery will be larger when N is small; but when N is large, the ratio for successful support recovery will converge slower to 100\% and oscillation exists. This can be observed in the zoomed-in part.
   2) When $\epsilon$ is chosen to be large, the ratio for successful support recovery will be smaller when N is small; but when N is large, the ratio for successful support recovery will go faster to 100\%  and no oscillation exists, see also the zoomed-in part in the figure.}
\label{fig.support}
\end{figure}

In practice, cross validation technique could be exploited to choose the tuning parameters. In the following, we will take the ADALASSO and the proposed method for granted to illustrate the idea and compare the performances for both methods when the parameters are obtained by the cross validation technique. In the ADALASSO algorithm and the proposed algorithm, there are two tuning parameters, namely $\gamma$ for the ADALASSO, and $\epsilon$ for the proposed method. In the following part, we will apply the 5-fold cross-validation method (see \cite{c18}) to choose the tuning parameters and then compare their performances based on the chosen tuning parameters. The procedure is as follows. At first, the tuning parameter is obtained by 5-fold cross validation, then it is applied to an independently generated test data which has the same dimension as the training data and the evaluation data. For different $N$, we run 100 i.i.d. realizations. In each realization, we record the value $\|\hat{\xv} - \xv^{0}\|_2^2$, where $\hat{\xv}$ denotes the estimate obtained by the estimator. $\epsilon$ are selected from $\{1/8,1/4,1/2\}$, $\gamma$ are selected from $\{1/2,1,2\}$ , and $N$ are chosen from $\{20,50,100,200,300,500\}$.  The results are reported in Fig. 4.

\begin{figure*}[thpb]
\centering
\label{fig.comp}
\subfigure[]{\includegraphics[width = 80mm]{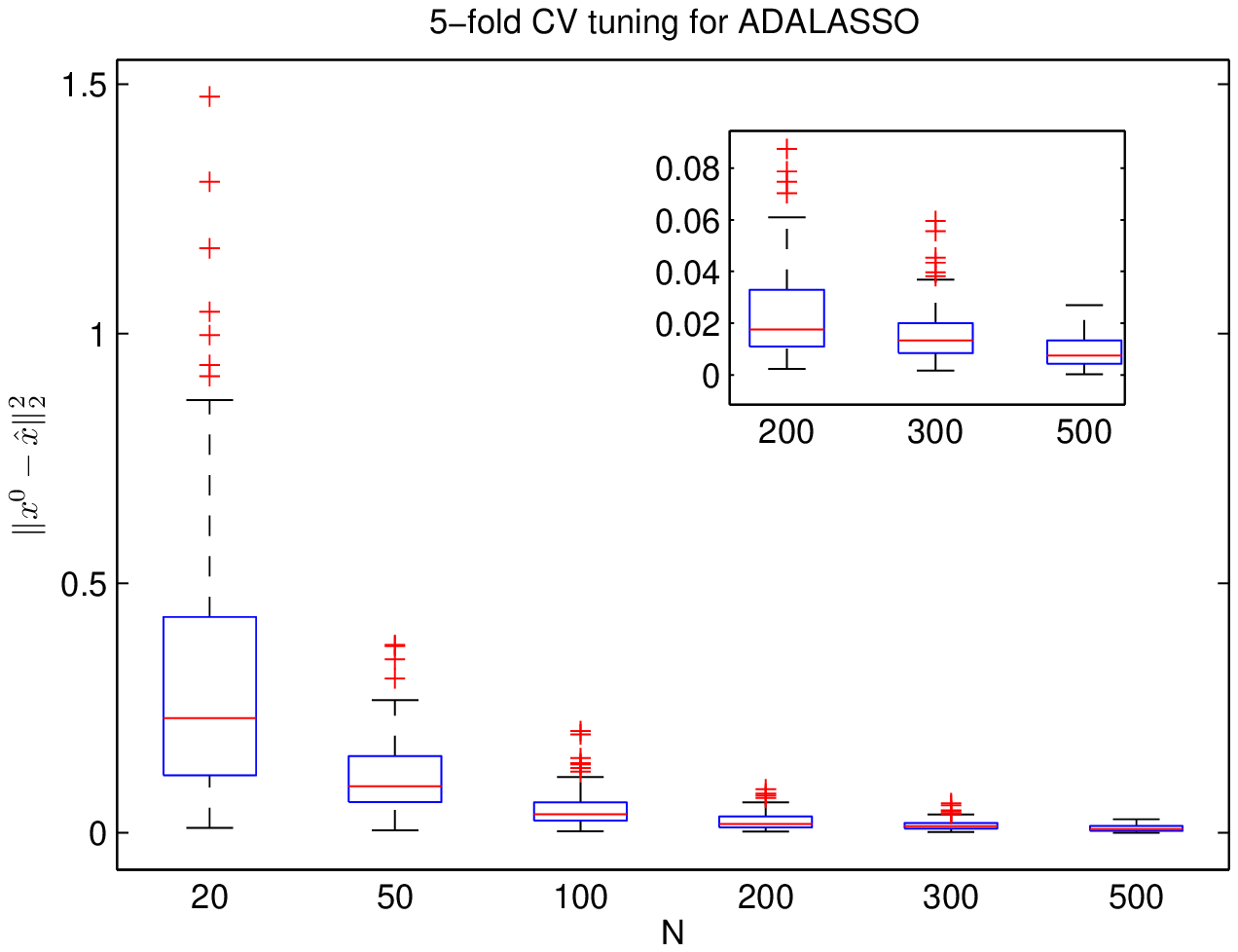}}
%  \caption{This figure demonstrates the boxplots of the recovery error obtained through the ADALASSO estimator by choosing the tuning parameter $\gamma$ with the 5-fold cross validation method. See also the discussions in the caption part of next figure.}
\subfigure[]{\includegraphics[width = 80mm]{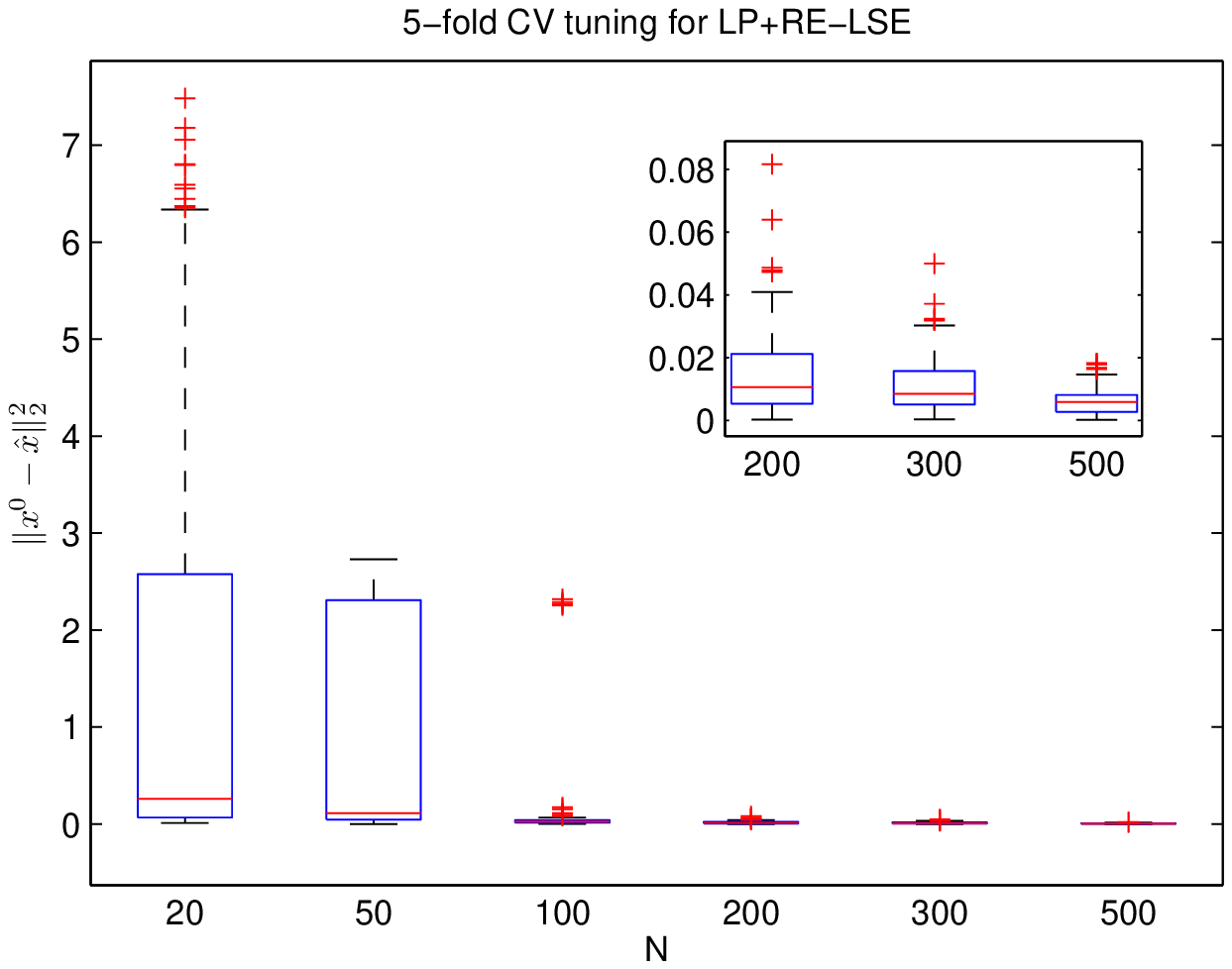}}
\caption{This figure demonstrates the boxplots of the recovery error obtained through the ADALASSO estimator and the proposed estimator when the tuning parameters are chosen by the 5-fold cross validation method. From this figure, we can see that performances of both methods are similar when $N$ is large, see the zoomed-in part in the figures. It can also be observed that when $N$ is small, the ADALASSO method has smaller recovery error compared with the proposed method.}
\end{figure*}
\end{subsection}

\begin{subsection}{Experiment 2}
In this part, we perform an experiment for recovering the sinusoids from noisy measurements.
The data is generated as follows:
\begin{equation}\nonumber
	y(t) = \sum_{k'=1}^{n'} c_{i_{k'}} \sin(w_{i_{k'}} t) + v(t).
	\label{eq.sin}
\end{equation}
Here both $\{w_{i_{k'}}\}_{k'}$ and $\{c_{i_{k'}}\}_{k'}$ are unknown, but we know that the frequencies do belong to a (larger, but of constant size) set $\{w_{k}\}_{k=1}^{n}$ of $n$ elements.
By sampling the system with period $t_s$, we obtain the system
\begin{equation}
	\mathbf{y} = \AM \mathbf{c}^{0} + \mathbf{v},
	\label{eq.sinsys}
\end{equation}	
where $\mathbf{y} = [y(t_s),\cdots, y(N t_s)]^{T}$.
The matrix $\AM \in \R^{N\times n}$ is defined as follows.
The $i$-th row of $\AM$ is given by
\begin{equation}
	\AM_i = \left[\sin(i w_1 t_s), \sin(i w_2 t_s),\dots, \sin(i w_n t_s)\right],
\end{equation}	
 for $i = 1,\cdots,N$. The parameter term and noise term are defined as
 $\mathbf{c}^{0} = [c_1,c_2,\cdots,c_n]^T$, and $\mathbf{v} =  [v(t_s), v(2 t_s),\cdots, v(N t_s)]^{T}$.

 In this experiment, $n=10$ and $\mathbf{c}^{0}=(1,1,1,0,\cdots,0)^T$,  $w_k = k$ for $k= 1,2,\cdots,n$. We increase $N$ up to 500 and the noise vector $\vv$ satisfies $\mathbf{v} \backsim \mathcal{N}(0,I_{N})$. We also assume that only the first three entries in $\{w_{k}\}_{k=1}^{n}$ occur effectively in the system of Eq. (\ref{eq.sin}) and the corresponding amplitudes are set to 1, i.e. $n'=3$ and $i_1 =1$, $i_2 =2$, $i_3 =3$. The sampling period $t_s$ is set to $0.1s$.

 The result using the proposed algorithm to recover $\xv^0$ is displayed in Fig. \ref{fig.exp2}. It is again clear that the proposed estimator is as efficient as the ORACLE estimator if one has enough samples.
 \begin{figure}[htbp]
\centering
\includegraphics[width = 80mm]{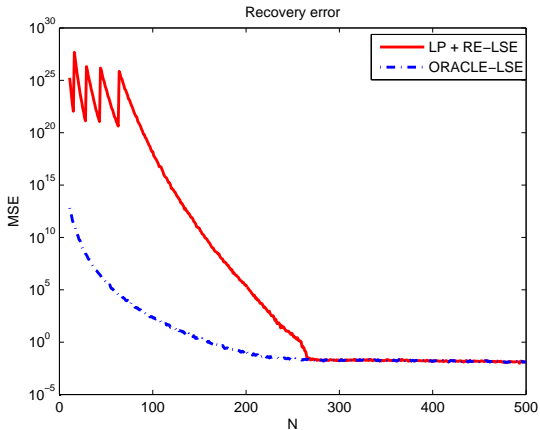}
  \caption{
  	Performance of applying the proposed estimator to recovery sinusoids functions from $N$  observations in Experiment 2.
  	This example also indicates that after a finite number the estimate is exactly equal to the ORACLE estimator.}
\label{fig.exp2}
\end{figure}

This is indeed predicted by the theory above since the $\AM$ in Eq. (\ref{eq.sinsys}) obeys the assumption of Eq. (\ref{eq.sing}).
This follows from the proposition given as:
\begin{Proposition}
   	There exist constants $\{C_{i,j}\}_{0\leq i,j\leq n}$ which do not depend on $N$, such that the following results hold.
	For any $1\le i\neq j \le n$, one has that:	
	\begin{equation}
		\left|(A^{T}A)_{i,j}\right|=\left|\sum_{t=1}^N \sin(t w_i t_s) \sin(t w_j t_s)\right| \leq C_{i,j}
		\label{eq.sin1}
	\end{equation}
	and for any $1\le i \le n$ that:
	\begin{equation}
		(A^{T}A)_{i,i}=\sum_{t=1}^N \left(\sin(t w_i t_s)\right)^2 \ge \frac{N}{2} - C_{i,i}.
		\label{eq.sin2}
	\end{equation}
\end{Proposition}
The proof is given in Appendix A.
With this proposition, an application of Ger\v{s}gorin circle theorem implies that the eigenvalues of $\AM^T\AM$ will increase with the order of $N$, which in turn implies Eq. (\ref{eq.sing}).
\end{subsection}
\end{section}
\begin{section}{Conclusion}
	This note presents an algorithm for solving an over-determined linear system from noisy observations,
	specializing to the case where the true 'parameter' vector is sparse.
	The proposed method does not need one to solve explicitly an optimization problem: it rather requires one to compute twice the LSE step, as well to perform a computationally cheap soft-thresholding step.
	Also, it is shown formally that the proposed method achieves the ORACLE property.
An open question is to quantify how many samples would be sufficient to guarantee exact recovery of $\xv^0$ for
given sparsity level $s$. In this note, we resort to the asymptotic Borel-Cantelli Lemma ('there exists such a number'),
but it is often of interest to have an explicit characterization of this number.
Another open question is how to find a suitable weighting matrix $\WM$ which can further improve the performance of the proposed algorithm.
\end{section}

\appendix
\section{Proof of Proposition 4}

\begin{proof}
The proof of (\ref{eq.sin1}) goes as follows. First
\begin{align}
	&\left|\sum_{t=1}^N \sin(t w_i t_s) \sin(t w_j t_s)\right| \nonumber \\
	&=\frac{1}{2} \left|\sum_{t=1}^N \left(\cos(t (w_i - w_j)t_s) - \cos(t (w_i + w_j)t_s)\right )\right| \nonumber \\
	& \le \frac{1}{2} \left|\sum_{t=1}^N \cos(t (w_i - w_j)t_s)\right| + \frac{1}{2} \left|\sum_{t=1}^N \cos(t (w_i + w_j)t_s)\right|.  \nonumber
\end{align}
We focus on bounding the term $\left|\sum_{t=1}^N \cos(t (w_i - w_j)t_s)\right|$,
the bound of the other term will follow along the same lines.

\begin{align} \nonumber
	&\left| \sum_{t=1}^N \cos(t (w_i - w_j)t_s) \right| \\ \nonumber
	&= \left| \operatorname{Re}\left(\frac{1 - e^{j(N+1) (w_i - w_j)t_s}}{1 - e^{j(w_i - w_j)t_s}}\right) -1 \right | \\ \nonumber
	&\le \left| \frac{1 - e^{j(N+1) (w_i - w_j)t_s}}{1 - e^{j(w_i - w_j)t_s}}\right | + 1 \\ \nonumber
	&\le \frac{2}{\left| 1 - e^{j(w_i - w_j)t_s}\right |} + 1, \nonumber
\end{align}
which is a constant which does not depend on $N$, so inequality (\ref{eq.sin1}) is obtained.

In order to prove inequality (\ref{eq.sin2}), observe that
\begin{align}
	\sum_{t=1}^N \left(\sin(t w_i t_s)\right)^2
		= \frac{1}{2}\sum_{t=1}^N \left(1 - \cos(2 t w_i t_s)\right) \nonumber \\
                  \ge \frac{N}{2} - \frac{1}{2} \left|\sum_{t=1}^N\cos(2 t w_i t_s) \right|. \nonumber
\end{align}
Using previous bounding method, $\frac{1}{2}\left|\sum_{t=1}^N\cos(2 t w_i t_s) \right|$ is also bounded by a constant $C_{i,i}$ which does not depend on $N$. This concludes the proof.
\end{proof}
\end{document}